\documentclass[journal]{IEEEtranTIE}
\usepackage{graphicx}
\pdfoutput=1
\usepackage{cite}
\usepackage{picinpar}
\usepackage{amsmath}
\usepackage{url}
\usepackage{flushend}
\usepackage[latin1]{inputenc}
\usepackage{colortbl}
\usepackage{soul}
\usepackage{multirow}
\usepackage{pifont}
\usepackage{color}
\usepackage{alltt}
\usepackage[hidelinks]{hyperref}
\usepackage{enumerate}
\usepackage{siunitx}
\usepackage{breakurl}
\usepackage{epstopdf}
\usepackage{pbox}
\usepackage{subfigure}
\usepackage{amsfonts}

\newtheorem{assumption}{Assumption}

\newtheorem{remark}{Remark}

\newtheorem{theorem}{Theorem}

\newenvironment{proof}{{\indent \indent \it Proof:}}

\begin{document}
\title{	Smooth Attitude Tracking Control of a 3-DOF Helicopter with Guaranteed Performance}

\author{
	\vskip 1em	
	Xidong Wang
	\thanks{
Xidong Wang is with the Research Institute of Intelligent Control and Systems, School of Astronautics, Harbin Institute of Technology, Harbin 150001, China (e-mail: 17b904039@stu.hit.edu.cn).
	}
}

\maketitle
	
\begin{abstract}
This paper presents a new prescribed performance control scheme for the attitude tracking of the three degree-of-freedom (3-DOF) helicopter system with lumped disturbances under mechanical constraints. First, a novel prescribed performance function is defined to guarantee that the tracking error performance has a small overshoot in the transient process and converges to an arbitrary small region within a predetermined time in the steady-state process without knowing the initial tracking error in advance. Then, based on the novel prescribed performance function, an error transformation combined with the smooth finite-time control method we proposed before is employed to drive the elevation and pitch angles to track given desired trajectories with guaranteed tracking performance. The theoretical analysis of finite-time Lyapunov stability indicates that the closed-loop system is fast finite-time uniformly ultimately boundedness. Finally, comparative experiment results illustrate the effectiveness and superiority of the proposed control scheme.
\end{abstract}


{}

\definecolor{limegreen}{rgb}{0.2, 0.8, 0.2}
\definecolor{forestgreen}{rgb}{0.13, 0.55, 0.13}
\definecolor{greenhtml}{rgb}{0.0, 0.5, 0.0}

\section{Introduction}

\IEEEPARstart{I}{n} recent years, the 3-DOF lab helicopter platform (Fig. 1) has received broad attention from researchers due to its similar dynamics with the real helicopter system \cite{Yang2020F}. In view of the high-performance attitude tracking control problem of the 3-DOF helicopter system, numerous advanced approaches have been proposed in recent years and validated via this experimental platform \cite{Yang2020F,2020arXiv,Kara2019F,Chen2018N,Zeghlache2017N,Castaneda2016F,Li2015F}. 
\begin{figure}[!t]\centering
	\includegraphics[width=8.5cm]{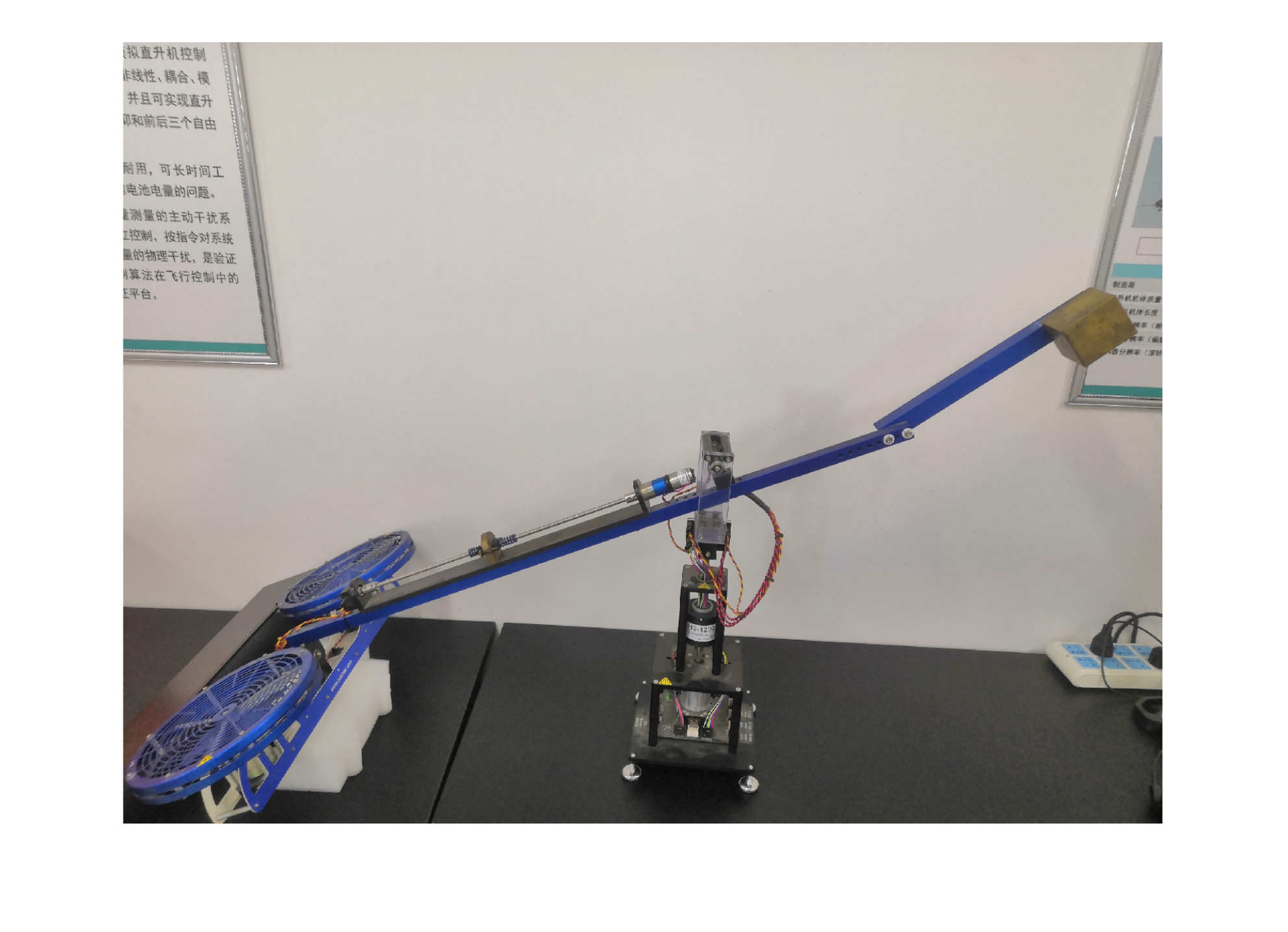}
	\caption{Structure of 3-DOF helicopter system under mechanical constraints}\label{FIG_1}
\end{figure}

Although the above-mentioned approaches obtain desired tracking performance, they only guarantee that the tracking error converges to a small region in the steady-state process. However, due to the mechanical constraints of the 3-DOF helicopter, the transient performance of attitude control needs to be concerned, embracing the overshoot and convergence rate of tracking errors. Recently, the prescribed performance control, first proposed by Bechlioulis and Rovithakis \cite{2008PPC}, has been extensively utilized to ensure both the transient and steady-state process within a prescribed performance. In \cite{2017overshootPPC}, a novel performance function is designed to tackle the tracking control problem with unknown initial errors and to satisfy the transient performance with small overshoot. The authors in \cite{2018finitePPC,2019finitePPC} propose a new concept named finite-time performance function, which ensures that the tracking error converges to a predefined small region within finite time. In \cite{PPC1,2019errorPPC}, two improved prescribed performance functions are present to restrain tracking errors within the prescribed envelops while precise initial errors are not required. The modified performance function is designed to address the tracking control problem of the quadrotor aircraft, which avoids singularity and ensures small overshoot during the transient process \cite{2020overshootPPC}.

Inspired by \cite{2018finitePPC,2020overshootPPC}, this paper aims to present a new prescribed performance control strategy for the attitude tracking of 3-DOF helicopter with lumped disturbance under mechanical constraints. A novel prescribed performance function is firstly designed to guarantee that the tracking error performance has a small overshoot in the transient process, and converges to an arbitrary small region within a predetermined time in the steady-state process. In addition, this performance function overcomes the demerit that the existing methods need to know the initial error in advance. Then, based on the novel performance function, the original tracking error system is transformed into an unrestrained system via error transformation. By means of the smooth finite-time control approach we proposed before, the presented control scheme achieves the smooth tracking of elevation and pitch angles with guaranteed tracking performance. 

The theoretical analysis of finite-time Lyapunov stability indicates that the closed-loop system is fast finite-time uniformly ultimately boundedness. The effectiveness as well as the superiority of our control scheme is validated by comparative numerical experiments.

The remainder of this paper is summarized as: In Section II, the model description and control objective of the 3-DOF helicopter system are given. A novel prescribed performance function is defined and an error transformation method is present in Section III. Section IV shows the controller design process and stability analysis of the closed-loop system. Section V presents comparative simulation result and discussion. The conclusion of this paper is provided in Section VI.

Notation: ${\mathop{\rm sign}\nolimits} ( \cdot )$ represents the standard signum function and the inverse function of $\Phi \left( z \right)$ is denoted by ${\Phi ^{ - 1}}\left( z \right)$.  
\section{Problem Formulation}
\subsection{Model description}
The specific structure of the 3-DOF helicopter system studied in this paper is described in [2]. Limited by the mechanical constraints, the operating domain of the helicopter system is defined as follows
\begin{equation}
\begin{aligned}
- {27.5^o} \le \alpha  \le  + {30^o}\\
 - {45^o} \le \beta  \le  + {45^o}
\end{aligned}
\end{equation}
where $\alpha $ and $\beta $ are elevation angle and pitch angle, respectively.

In terms of the derivation of [2], the unified tracking error model of the elevation and pitch channel can be formulated as follows
\begin{equation}
\begin{aligned}
&{\dot e_i} = {e_{i + 1}}\\
&{\dot e_{i + 1}} = {g_i}{v_i} + {f_i} + {d_i}
\end{aligned}
\end{equation}
where $i = 1,2$. ${e_1} = \alpha - {x_{\alpha d}}(t),{e_3} = \beta - {x_{\beta d}}(t)$ represent the tracking errors of the elevation and pitch angle respectively. The definitions and values of other symbols can be found in [2].
\subsection{Control objective}
The control objective is to design the smooth finite-time controllers such that the elevation and pitch angles can track the given desired trajectories respectively within constrained errors.
\begin{remark}
As the elevation and pitch angles are limited by the mechanical structure during the entire control process, it is necessary to restrict the tracking errors, especially in the transient process.
\end{remark}
\section{Prescribed Performance Function and Error Transformation}
\subsection{Prescribed performance function}
Motivated by \cite{2018finitePPC,2020overshootPPC}, we develop a novel prescribed performance function, which is defined as follows
\begin{equation}
{P_u}\left( t \right) =
{\begin{cases}
\left( {e\left( 0 \right) + \delta  - {\lambda _\infty }} \right)\frac{{{T_f} - t}}{{{T_f}}}{e^{\left( {1 - \frac{{{T_f}}}{{{T_f} - t}}} \right)}} + {\lambda _\infty },t \in \left[ {0,{T_f}} \right) \hfill \\
{\lambda _\infty },t \in \left[ {{T_f}, + \infty } \right)\hfill \\
\end{cases} }
\end{equation}
\begin{equation}
{P_l}\left( t \right) =
{\begin{cases}
\left( {e\left( 0 \right) - \delta  + {\lambda _\infty }} \right)\frac{{{T_f} - t}}{{{T_f}}}{e^{\left( {1 - \frac{{{T_f}}}{{{T_f} - t}}} \right)}} - {\lambda _\infty },t \in \left[ {0,{T_f}} \right) \hfill \\
-{\lambda _\infty },t \in \left[ {{T_f}, + \infty } \right)\hfill \\
\end{cases} }
\end{equation}
where $\delta > 0,{\lambda _\infty }> 0,{T_f}> 0$ are design parameters.

Then, the following inequality is introduced to guarantee the prescribed performance of the tracking error.
\begin{equation}
{P_l}\left( t \right) < e\left( t \right) < {P_u}\left( t \right)
\end{equation}
\begin{figure}[!t]\centering
	\includegraphics[width=0.5\textwidth]{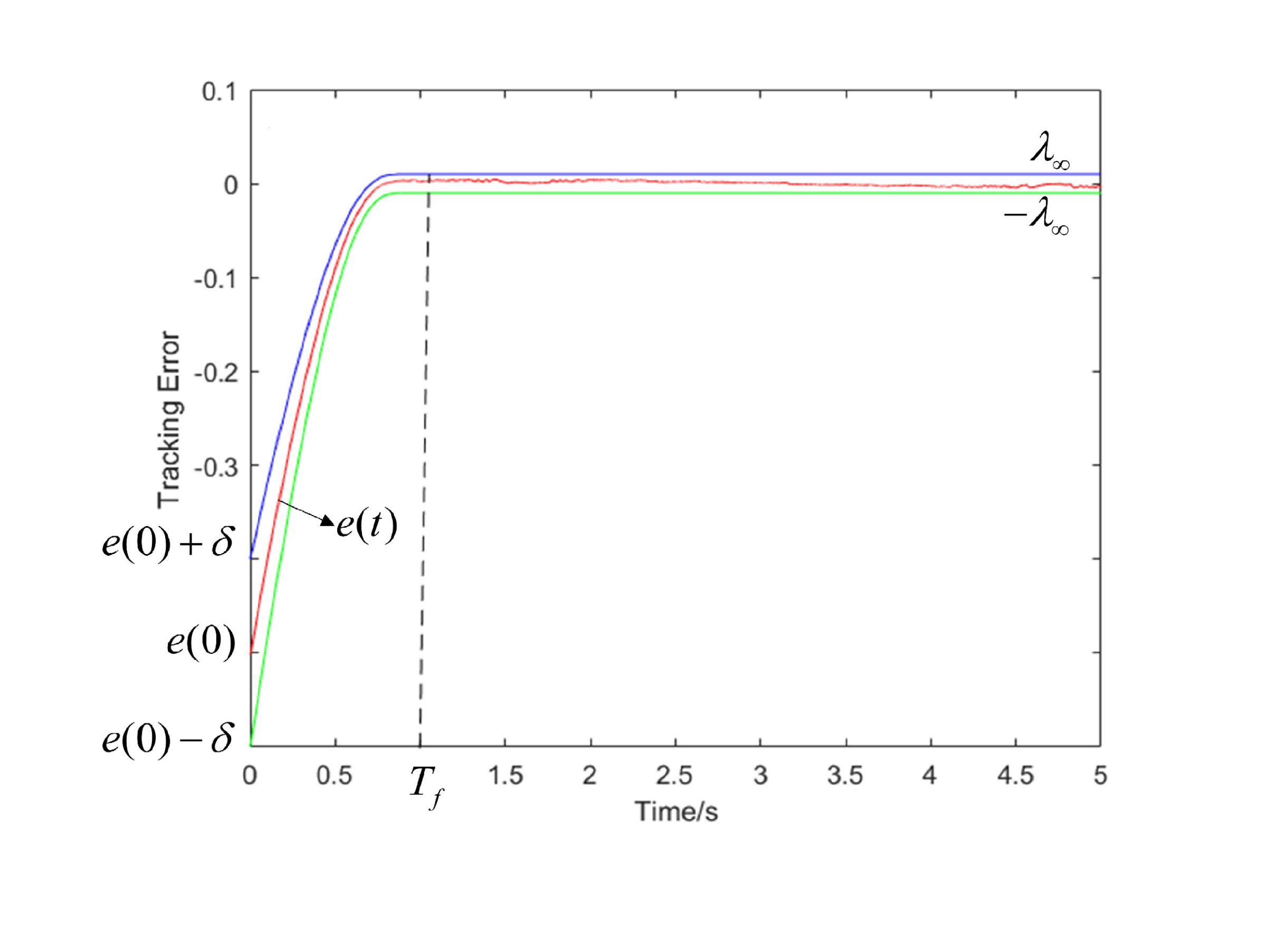}
	\caption{Schematic of the novel designed prescribed performance (5)}
\end{figure}

Fig. 2 presents the novel designed prescribed performance (5). It is demonstrated from Fig. 2, (3) and (4) that the tracking error performance has a small overshoot as well as converges to an arbitrary small range $\left( { - {\lambda _\infty },{\lambda _\infty }} \right)$ within finite time where the settling time ${T_f}$ is independent of the initial condition and can be preset by the users. Moreover, this performance function overcomes the demerit that the traditional PPC [8] needs to know the precise initial error in advance.
\subsection{Error transformation}
The purpose of error transformation is to transform the constrained tracking error system into unconstrained system. Then, the controller, designed to achieve the unconstrained system bounded, can guarantee that the original tracking error converges to the prescribed envelope. The error transformation for this paper is designed as follows
\begin{equation}
e\left( t \right) = {P_u}\left( t \right)\Phi \left( {z\left( t \right)} \right) + {P_l}\left( t \right)\left( {1 - \Phi \left( {z\left( t \right)} \right)} \right)
\end{equation}
where $z\left( t \right)$ is the transformed error and $\Phi \left( {z\left( t \right)} \right)$ is designed as
\begin{equation}
\Phi \left( {z\left( t \right)} \right) = \frac{1}{\pi }\arctan \left( {z\left( t \right)} \right) + \frac{1}{2}
\end{equation}
\begin{theorem}
If the transformed error $z\left( t \right)$ is bounded, then $e\left( t \right)$ will always be constrained to the prescribed envelope defined by (5) for all $t \ge 0$.
\end{theorem}
\begin{proof}
By the boundedness of $z\left( t \right)$, we have that there exists $\eta  \in \left( {0, + \infty } \right)$ such that $ - \eta  \le z\left( t \right) \le \eta $.

Then we have the following estimates
\begin{equation}\small
0 < \frac{1}{\pi }\arctan \left( { - \eta } \right) + \frac{1}{2} \le \Phi \left( {z\left( t \right)} \right) \le \frac{1}{\pi }\arctan \left( \eta  \right) + \frac{1}{2} < 1
\end{equation}

Note that (6) can be rewritten as
\begin{equation}
e\left( t \right) = \left( {{P_u}\left( t \right) - {P_l}\left( t \right)} \right)\Phi \left( {z\left( t \right)} \right) + {P_l}\left( t \right)
\end{equation}

Applying the range of $\Phi \left( {z\left( t \right)} \right)$, we obtain
\begin{equation}
{P_l}\left( t \right) < e\left( t \right) < {P_u}\left( t \right)
\end{equation}

The proof is completed.
\end{proof}
\section{Controller Design}
This section takes the elevation channel as an example to expound the design procedure of the controller. Following a similar process, the controller for the pitch channel can also be designed.

Considering the tracking error system of the elevation channel
\begin{equation}
\begin{aligned}
&{\dot e_1} = {e_2}\\
&{\dot e_2} = {g_1}{v_1} + {f_1} + {d_1}
\end{aligned}
\end{equation}

Define ${z_1}\left( t \right)$ as the transformed error of ${e_1}\left( t \right)$, by adopting the error transformation (6), one can obtain
\begin{equation}
\begin{aligned}
&{{\dot z}_1} = {z_2}\\
&{{\dot z}_2} = {M_1} + {M_2} + {M_3}\left( {{g_1}{v_1} + {f_1} + {d_1}} \right)
\end{aligned}
\end{equation}
where
\begin{equation}
{M_3} = \frac{{d{\Phi ^{ - 1}}\left( {{z_1}\left( t \right)} \right)}}{{d{z_1}}} \cdot \frac{1}{{{P_u} - {P_l}}}
\end{equation}
with ${M_3} > 0$ and $\Phi \left( {{z_1}\left( t \right)} \right),{P_u},{P_l}$ defined in (7), (3), (4).

For designing the controller, we make the following assumption.
\begin{assumption}
There exists unknown positive constants $\mu $ such that 
\begin{equation}
\left| {\frac{d}{{dt}}\left( {{M_3} \cdot {d_1}} \right)} \right| \le \mu 
\end{equation}
\end{assumption}

A non-singular integral sliding mode surface \cite{2002Surface} is defined as follows
\begin{equation}
s = {z_2} + \int_0^t {{z_s}d\tau } 
\end{equation}
where ${z_s} = {\gamma _1}{\left| {{z_1}} \right|^p}{\mathop{\rm sgn}} ({z_1}) + {\gamma _2}{\left| {{z_2}} \right|^{2p/(1 + p)}}{\mathop{\rm sgn}} ({z_2})$ and ${\gamma _1} > 0,{\gamma _2} > 0,p \in (0,1)$. 

Based on the smooth finite-time control approach we proposed before [2], the elevation channel controller is designed as
\begin{equation}
\begin{aligned}
{v_1} = &g_1^{ - 1}\left( {M_3^{ - 1}{w_1} - {f_1}} \right)\\
{w_1} =  &- {L_1}(t){\left| s \right|^{\frac{{m - 1}}{m}}}{\mathop{\rm sgn}} (s) - {L_2}(t)s - z - {M_1} \\
 &- {M_2}- \int_0^t {\left( {{L_3}(t){{\left| s \right|}^{\frac{{m - 2}}{m}}}{\mathop{\rm sgn}} (s) + {L_4}(t)s} \right)d\tau } 
\end{aligned}
\end{equation}
where the formulation of ${L_1}(t),{L_2}(t),{L_3}(t),{L_4}(t)$ and the range of control parameters can be found in [2].
\begin{theorem}
Consider the transformed error system (12) with lumped disturbance satisfying  \emph{Assumption 1}, under the proposed control law (16) with the sliding mode surface defined in (15). Then the following conclusions can be obtained: 1) The transformed tracking error is fast finite-time uniformly ultimately boundedness. 2) The original tracking error lies in the prescribed envelope constrained by (5). 
\end{theorem}
\begin{proof}
By applying the \emph{Proposition 1} in [2], the conclusion 1) can be easily proved. The conclusion 1) shows that the transformed tracking error is bounded. Then applying \emph{Theorem 1}, we can further obtain that the original tracking error will always be constrained to the prescribed envelope defined by (5) for all $t \ge 0$. The proof is completed.
\end{proof}
\section{Numerical Simulations}
In this section, we take the simulation experiment of elevation angel tracking as an example to expound the effectiveness of our proposed control scheme.
\subsection{Case 1: Attitude tracking control with guaranteed performance}
In Case 1, the initial elevation angle of the 3-DOF helicopter system is $-24^o$. The lumped disturbance is assumed as ${d_1}(t) = 0.2\sin (t)$. The desired trajectory of elevation angle is given as:
\begin{equation}
{x_{\alpha d}}(t){\rm{ = }}0.2\sin (0.08t - \frac{\pi }{2})
\end{equation}

The parameters of our proposed PPF are set as ${T_f} = 1.5,\delta  = 0.1,{\lambda _\infty } = 0.01$, while other parameters of the control law (16) can be found in [2].
\begin{figure}[!t]\centering
	\includegraphics[width=8.5cm]{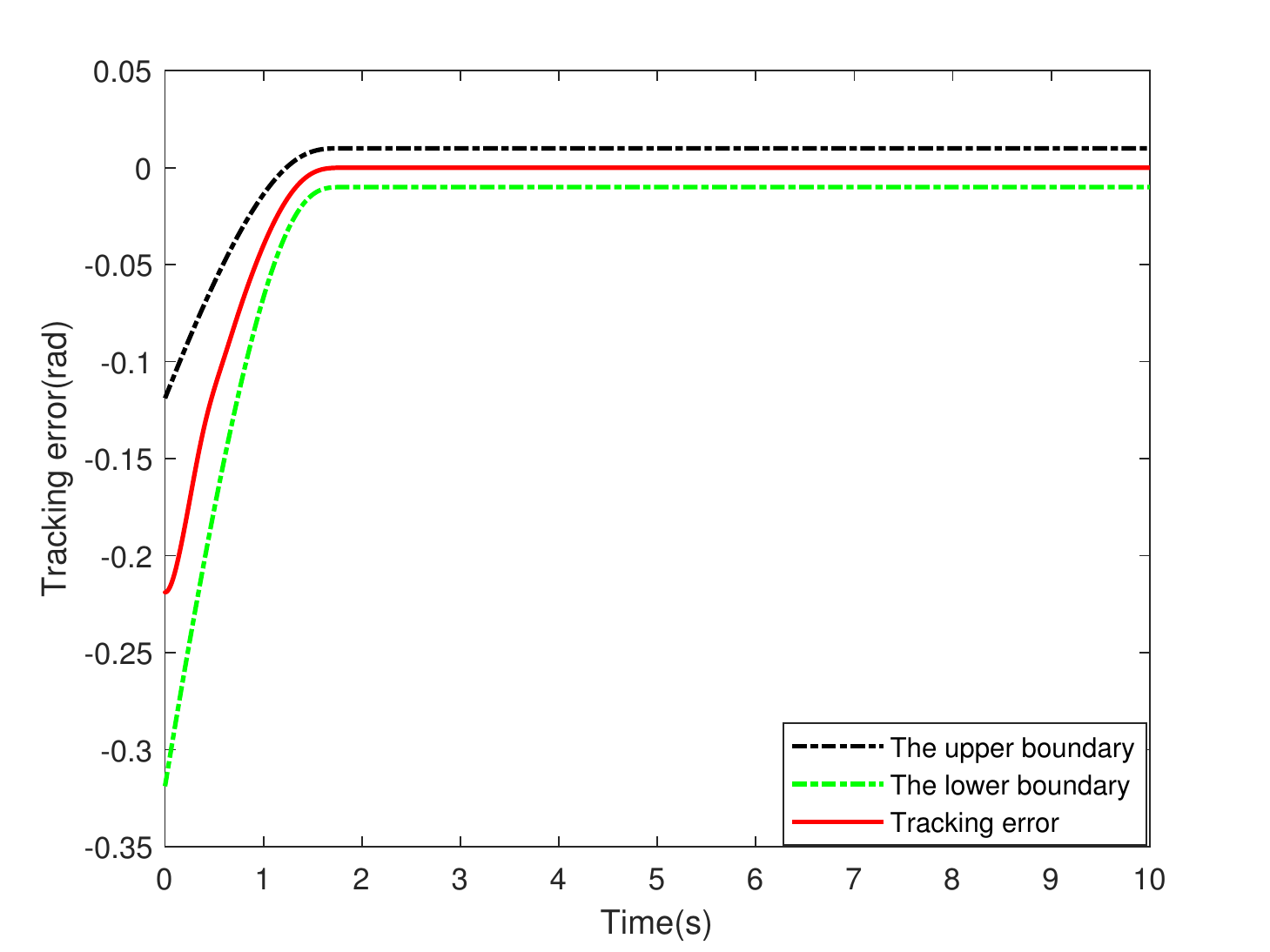}
	\caption{The curves of tracking error and boundary}
\end{figure}

Fig. 3 illustrates the response of the tracking error by utilizing our proposed control scheme, which shows that the tracking error of elevation angel can be kept within the prescribed envelope defined by (5). Moreover, the tracking error performance has a small overshoot in the transient process, and converges to $\left( { - 0.01,0.01} \right)$ within our preset $1.5s$ in the steady-state process without knowing the initial tracking error in advance. 
\subsection{Case 2: Contrast experiment}
In Case 2, the initial conditions of the 3-DOF helicopter, the desired trajectory, and the parameters of the proposed control scheme are set to be the same as in Case 1.

For comparison, the control scheme proposed in \cite{Tran2020} is employed, which adopts the PPF: $\rho \left( t \right) = \left( {{\rho _0} - {\rho _\infty }} \right){e^{ - kt}} + {\rho _\infty }$ with ${\rho _0}{\rm{ = 0}}{\rm{.48,}}{\rho _\infty } = 0.01,k = 2$ and utilizes the error transformation: ${e_1}\left( t \right) = \rho \left( t \right)\Theta \left( {{z_1}} \right)$,where
\begin{equation}
\Theta \left( {{z_1}} \right){\rm{ = }}\frac{{{e^{{z_1}}} - {e^{ - {z_1}}}}}{{{e^{{z_1}}} + {e^{ - {z_1}}}}}
\end{equation}

The other parameters of the control scheme proposed in [16] can be found in [2].
\begin{figure}
	\centering
	\subfigure[The curves of tracking error]{
	\includegraphics[width=0.4\textwidth]{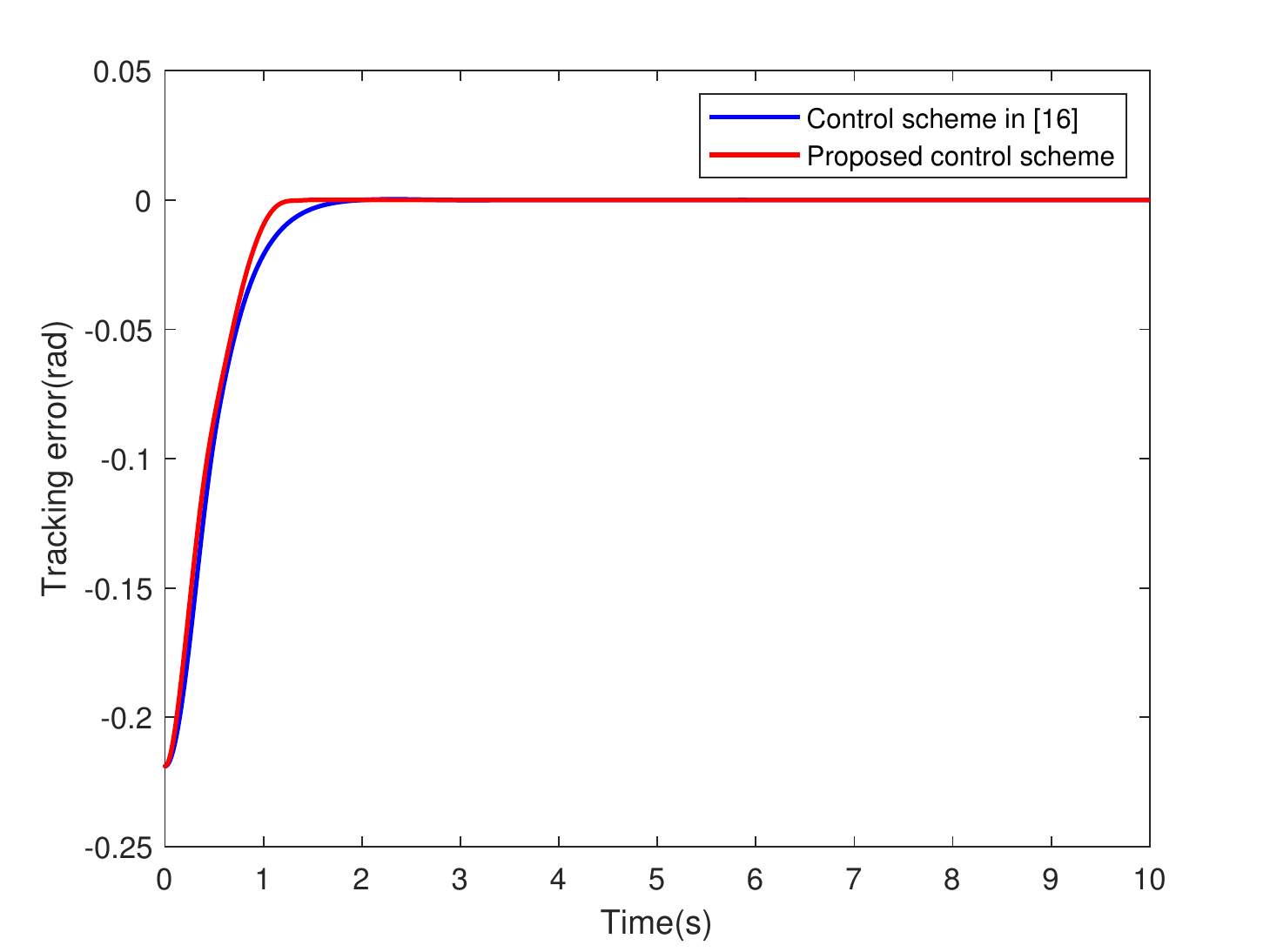}}
	\subfigure[Partial enlarged graph of tracking error transient response]{
	\includegraphics[width=0.4\textwidth]{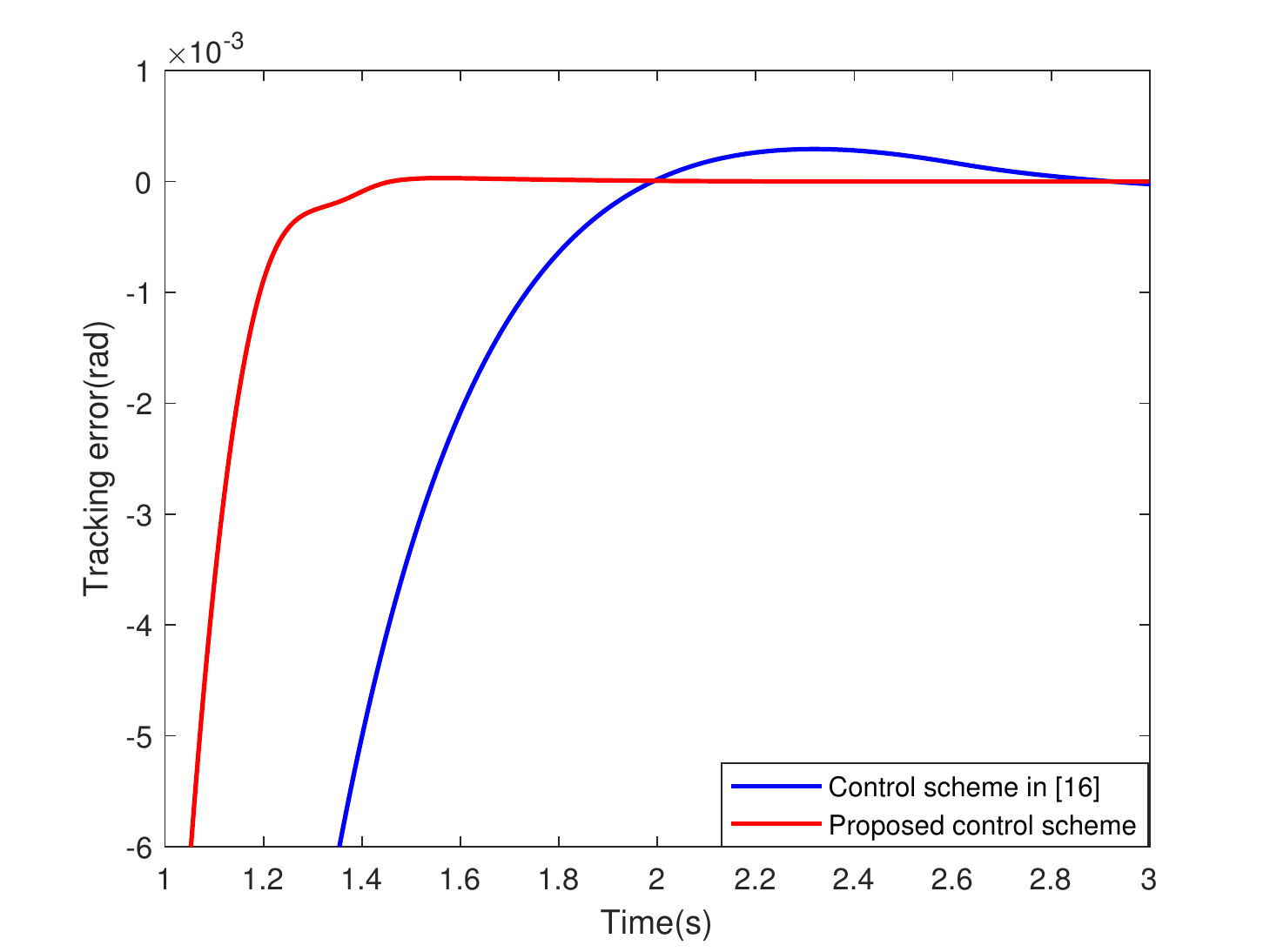}}
	\subfigure[Partial enlarged graph of tracking error steady-state response]{
	\includegraphics[width=0.4\textwidth]{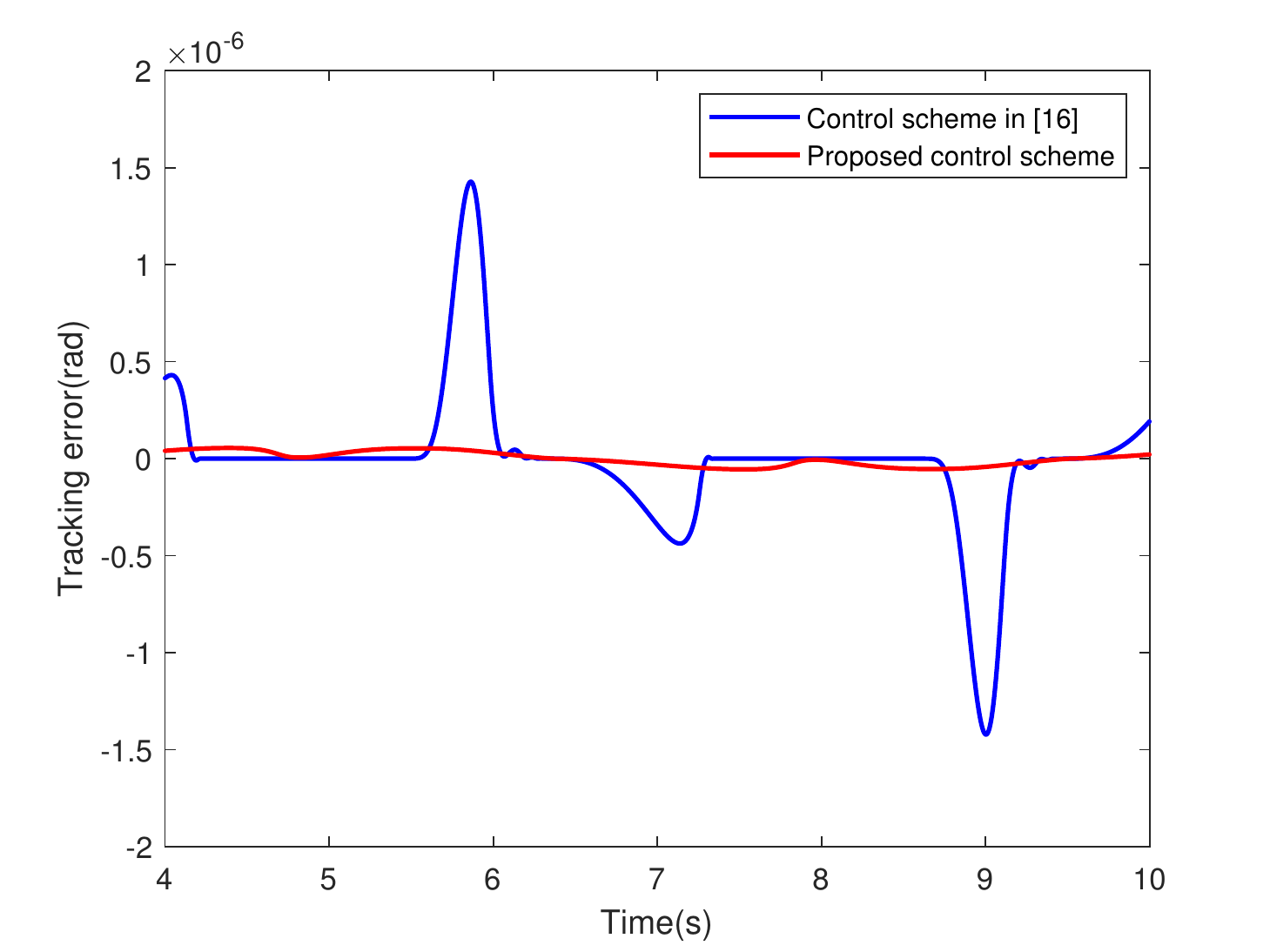}}
	\caption{Results of Case 2}
\end{figure}

Fig. 4: (a) presents the response of the tracking error by utilizing our proposed control scheme and the method in [16]. Fig. 4: (b) demonstrates the partial enlargement of tracking error in transient process. Fig. 4: (c) illustrates the partial enlargement of tracking error in steady-state process. Fig. 4: (a) and (b) show that the tracking error performance of our proposed control scheme has a smaller overshoot than that in [16], and can converge to an arbitrary small region within a predetermined time without knowing the initial tracking error in advance. In addition, it can be seen from Fig. 4 (c) that our proposed control scheme provides more accurate and smoother output for elevation tracking control.
\section{Conclusion}
In this paper, a novel prescribed performance control scheme has been proposed to cope with the attitude tracking control problem of a 3-DOF helicopter system with lumped disturbances under mechanical constraints. First, a novel prescribed performance function is designed to ensure that the tracking error performance has a small overshoot in the transient process, and converges to an arbitrary small range within a preset time in the steady-state process without knowing the initial tracking error in advance. Then, based on the novel prescribed performance function and error transformation, the predefined performance tracking error system is transformed into an unrestrained error system. With the aid of the smooth finite-time control approach we proposed before, the presented control scheme achieves the smooth tracking of elevation and pitch angles with guaranteed tracking performance. By using finite-time Lyapunov stability theory, the closed-loop system is proven to be fast finite-time uniformly ultimately boundedness. The superiority of the presented approach is validated by comparative numerical simulations.

\bibliographystyle{Bibliography/IEEEtranTIE}
\bibliography{Bibliography/IEEEabrv,myRef}\ 

\end{document}